\documentclass[
  aps,
  prl,
  reprint,               
  superscriptaddress,    
  longbibliography
]{revtex4-2}

\usepackage{amsthm}
\usepackage{amsmath}
\usepackage{latexsym}
\usepackage{amsfonts}
\usepackage{amssymb}
\usepackage{color}
\usepackage{bbm,dsfont}
\usepackage{graphicx}
\usepackage{hyperref}
\usepackage{subfigure}
\usepackage{cancel}

\usepackage{soul}

\newtheorem{proposition}{Proposition}
\newtheorem{proposition?}{Proposition?}
\newtheorem{theorem}{Theorem}

\theoremstyle{definition}

\definecolor{agcol}{rgb}{0.9, 1.0, 0.9}
\definecolor{jkcol}{rgb}{0.5, 0.9, 0.8}

\newcommand{\bra}[1]{\left\langle#1\right\rvert}
\newcommand{\ket}[1]{\left\lvert#1\right\rangle}

\begin{document}
\title[]{Characterizing Fisher information of quantum measurement}

\author{Rakesh Saini}
\email[]{rakesh7698saini@gmail.com}
\affiliation{Centre for Engineered Quantum Systems, Department of Physics and Astronomy, Macquarie University, Sydney NSW 2113, Australia}
\author{Jukka Kiukas}
\affiliation{Department of Mathematics, Aberystwyth University, Aberystwyth SY23 3BZ, United Kingdom}
\author{Daniel Burgarth}
\affiliation{Physics Department, Friedrich-Alexander Universit\"at of Erlangen-Nuremberg, Staudtstr. 7, 91058 Erlangen, Germany}
\affiliation{Department of Mathematics, Aberystwyth University, Aberystwyth SY23 3BZ, United Kingdom}
\affiliation{Centre for Engineered Quantum Systems, Department of Physics and Astronomy, Macquarie University, Sydney NSW 2113, Australia}
\author{Alexei Gilchrist}
\affiliation{Centre for Engineered Quantum Systems, Department of Physics and Astronomy, Macquarie University, Sydney NSW 2113, Australia}

\date{\today}

\begin{abstract} 
Informationally complete measurements form the foundation of universal quantum state reconstruction, while quantum parameter estimation is based on the local structure of the manifold of quantum states. Here we establish a general link between these two aspects, in the context of a single informationally complete measurement, by employing a suitably adapted operator frame theory. In particular, we bound the ratio between the classical and quantum Fisher information in terms of the spectral decomposition of the associated frame operator, and connect these bounds to the optimal and least optimal directions for parameter encoding. The geometric and operational characterization of information extraction thus obtained reveals the fundamental tradeoff imposed by informational completeness on local quantum parameter estimation.
\end{abstract}
\maketitle

\emph{Quantum state tomography} plays a central role in quantum information science, underpinning key applications in quantum sensing, communication, and metrology~\cite{2004Gio336,1969Hel252,2009Par137,2011Gio229,2014Tot006,2016Szc639,2017Deg002,2020Liu001}.
A crucial factor determining the efficiency of tomography—or more specifically, the accuracy of parameter estimation—is the choice of quantum measurement, mathematically described by a \emph{positive operator-valued measure} (POVM)~\cite{1996Buschquantum,2011Hol001}.
For a given parameter-encoded quantum state and measurement, the amount of information that can be extracted about the parameter is quantified by the \emph{classical Fisher information}.

If the parameter is encoded along a specific direction in the state space, one can, in principle, optimize over all possible measurements to find the POVM that extracts the maximal information about that parameter and its encoding. The corresponding ultimate limit is set by the \emph{quantum Fisher information}, which directly determines the \emph{quantum Cramér–Rao bound}—the fundamental lower bound on the variance (mean-squared error) of any unbiased estimator~\cite{2009Par137}.
This bound is achieved when the measurement consists of projectors onto the eigenvectors of the \emph{symmetric logarithmic derivative} (SLD), which is defined by the parametric derivative of the quantum state.

Importantly, the \emph{quantum Fisher information} depends only on the specific quantum state and the encoding of specific parameters. Thus, the optimal POVM for one parameterized family of states may be completely uninformative for another. To overcome this limitation, it is natural to consider \emph{informationally complete POVMs} (IC-POVMs)—measurements whose elements span the entire space of linear operators. IC-POVMs are particularly valuable because they allow estimation of parameters encoded in \emph{any} direction, given sufficient measurement statistics, making them ideally suited for universal parameter estimation tasks~\cite{2004Dar487}.

Because IC-POVMs form spanning sets, their mathematical structure is naturally described by \emph{frame theory}, which generalizes the concept of a basis to redundant spanning sets~\cite{2018waldronAITFTF}.
Frame theory provides a unified language to analyze informational completeness, reconstruction stability, and robustness of quantum measurements.

In this Letter, we present a unified framework connecting local parameter estimation, universal state reconstruction, Fisher information, and frame theory. We show that for an IC-POVM and a full-rank quantum state—whether with no prior information (maximally mixed) or with prior information (a specific state)—the ratio between classical and quantum Fisher information is fully characterized by the spectral properties of the corresponding \emph{frame operator}. Since IC-POVMs necessarily fail to saturate the quantum Cramér–Rao bound, this ratio is always strictly less than one. We derive tight upper and lower bounds on this ratio, determined by the extremal eigenvalues of the frame operator, and identify the associated eigenvectors, which correspond to the \emph{best} and \emph{worst} directions for parameter encoding.

Recent work on \emph{universally Fisher-symmetric informationally complete} measurements has shown that certain IC-POVMs are equally informative in all parameter directions~\cite{2018zhu404,2016Lin402}.
Our results encompass these as special cases, providing a broader and more general framework. In particular, the universal Fisher-symmetry condition emerges naturally as a limiting case within our spectral characterization. The broader goal of our work, however, is to provide a complete geometric and operational characterization of the information-extraction capabilities of any informationally complete measurements and establish fundamental limits on their optimal performance.

\emph{Local quantum estimation.}---
Consider a finite-dimensional Hilbert space and a local quantum statistical model $\theta \mapsto \rho_\theta$ defined in a neighborhood of $\theta = 0$, where the reference state $\rho_0 = \rho$ has full rank and is known. 
The direction of parameter change is given by the \emph{symmetric logarithmic derivative} (SLD) $L$, defined by
\begin{equation}
\frac{d\rho_\theta}{d\theta}\Big|_{\theta = 0} = \frac{1}{2}(L\rho + \rho L).
\end{equation}

When estimating $\theta$ using a single POVM $\mathsf{E} = \{E_i\}_{i \in I}$, the corresponding \emph{classical Fisher information} quantifying the optimal estimator precision is
\begin{equation}
I_C[\mathsf{E}, \rho] 
= \sum_{i \in I} \frac{\big[\mathrm{Re}\,\mathrm{Tr}(\rho L E_i)\big]^2}{\mathrm{Tr}(\rho E_i)}.
\end{equation}
Maximizing this quantity over all POVMs yields the \emph{quantum Fisher information} \(I_Q[\rho] = \mathrm{Tr}[\rho L^2]\),
which determines the ultimate achievable precision through the \emph{quantum Cramér--Rao bound}~\cite{1969Hel252,2011Hol001}:
\begin{equation}
\mathrm{Var}(\hat{\theta}) \ge \frac{1}{I_C[\mathsf{E}, \rho]} \ge \frac{1}{I_Q[\rho]},
\end{equation}
valid for any unbiased estimator $\hat{\theta}$. Expanding the state $\rho_\theta$ locally around \(\theta=0\) gives
\begin{equation}
\rho_\theta = \rho + \frac{\theta}{2}(L\rho + \rho L) + O(\theta^2),
\end{equation}
and since $\rho$ is full rank, the first order approximation of $\rho_\theta$ remains positive and trace one in a neighborhood of $\theta = 0$, reflecting the fact that the SLD completely determines the local model around the reference state $\rho$.

\emph{Frame-theoretic formulation.}---
To connect local estimation with the structure of quantum measurements, we equip the real vector space $\mathcal{M}$ of Hermitian operators with the \emph{state-dependent inner product}
\(
\langle A, B \rangle_\rho = \mathrm{Re}\,\mathrm{Tr}[\rho AB].
\)
Given an IC-POVM $\mathsf{E} = \{E_i\}_{i\in I}$, we define a linear map $V : \mathbb{R}^{|I|} \to \mathcal{M}$ by
\begin{equation}\label{frame}
V e_i = \langle E_i, \mathbb{I} \rangle_\rho^{-1/2} E_i, \qquad i \in I
\end{equation}
where \(e_i\) is the standard basis, and denote its adjoint (with respect to the above $\rho$-dependent inner product) by $V^{\dagger}$. Since $\mathsf E$ is informationally complete, the operators in \eqref{frame} span $\mathcal M$, and hence constitute a \emph{frame} \cite{2018waldronAITFTF}, with the associated \emph{frame operator}
\begin{equation}
\mathcal{F} := VV^{\dagger}.
\end{equation}
We stress that the frame operator \(\mathcal{F}\) depends on both the state \(\rho\) and the IC-POVM \(\mathsf{E}\).
Its action is given by
\begin{equation}
\langle A,\mathcal{F}(B) \rangle_\rho 
= \sum_{i \in I} \frac{\langle A, E_i \rangle_\rho\,\langle E_i, B \rangle_\rho}{\langle E_i, \mathbb{I} \rangle_\rho}
\end{equation}
for all $A, B \in \mathcal{M}$. Crucially, if $L$ is the SLD of some local parametric model around $\rho$, we obtain
\begin{equation}
I_C[\mathsf{E}, \rho] = \langle L,\mathcal{F}(L) \rangle_\rho,
\end{equation}
while $I_Q[\rho] = \langle L, L \rangle_\rho$. Thus, the ratio $I_C/I_Q$ can be expressed in terms of $\mathcal{F}$.

\emph{Frame operator properties.}---
We now summarize the essential spectral properties of the frame operator $\mathcal{F}$.

\begin{proposition}\label{framebounds}
Let $\rho$ be a full-rank quantum state and $\mathsf{E}$ an IC-POVM. Then the following hold:
\begin{enumerate}
\item[(a)] $\mathcal F$ is a positive Hermitian operator in the real Hilbert space $\mathcal M$ with inner product $\langle \cdot,\cdot\rangle_\rho$.
\item[(b)] All eigenvalues of $\mathcal{F}$ lie within the interval $[0,1]$.
\item[(c)] The top eigenvalue equals $1$, and $\mathbb I$ is the only corresponding eigenvector.
\end{enumerate}
\end{proposition}

\begin{proof} It is easy to check that $\mathcal F$ is Hermitian. Now fix $A \in \mathcal{M}$. For any positive semidefinite $E \in \mathcal{M}$,
\begin{align}\label{cauchy}
\langle A, E\rangle_\rho^2
&= (\mathrm{Re}\,\mathrm{Tr}[\rho AE])^2 \nonumber\\
&\le|{\rm Tr}[\rho AE]|^2
= |{\rm Tr}[\sqrt{\rho}A\sqrt{E}\sqrt{E}\sqrt{\rho}]|^2 \nonumber \\
&\le {\rm Tr}[\sqrt{\rho}E\sqrt{\rho}]\,{\rm Tr}[\sqrt{\rho}AEA\sqrt{\rho}] \nonumber\\
&= \langle E, \mathbb{I}\rangle_\rho\,\langle AEA, \mathbb{I}\rangle_\rho
\end{align}
by the Cauchy--Schwarz inequality. Hence,
\begin{align}\label{Fbound}
\langle A,\mathcal{F}(A) \rangle_\rho 
&= \sum_{i \in I} \frac{\langle A, E_i\rangle_\rho^2}{\langle E_i, \mathbb{I}\rangle_\rho}
\le \sum_{i \in I} \langle A E_i A, \mathbb{I}\rangle_\rho
= \langle A, A\rangle_\rho,
\end{align}
so all eigenvalues of $\mathcal{F}$ lie in $[0,1]$. Furthermore, since $\langle \mathbb{I},\mathcal{F}(\mathbb{I})\rangle_\rho = \langle \mathbb{I}, \mathbb{I}\rangle_\rho$, the operator $\mathbb{I}$ is an eigenvector with eigenvalue $1$. Furthermore, $\mathbb{I}$ is the \emph{only} eigenvector corresponding to $\lambda = 1$ (shown in Supplemental Material for completeness).  
 
\end{proof}

\begin{proposition}
Each eigenvalue $\lambda \neq 1$ of the frame operator $\mathcal{F}$ corresponds to the Fisher information of some local quantum model $\rho_\theta$ around $\rho$, and the corresponding eigenvectors lie in $\mathcal M_0^\rho$ (defined below).
\end{proposition}

\begin{proof}
For any local model $\rho_\theta$, it follows from the definition that the SLD $L$ must satisfy ${\rm tr}[\rho L]=0$.

Therefore, any direction $L$ valid for parameter encoding must belong to the ``mean-zero'' subspace
\begin{equation}
\mathcal M_0^\rho := 
\{\, X \in \mathcal M \;|\; 
\langle \mathbb I, X \rangle_\rho = \mathrm{Tr}[\rho X] = 0 \,\}
= \{\mathbb I\}^\perp,
\end{equation}
where the orthogonal complement is with respect to $\langle \cdot , \cdot \rangle_\rho$.
Conversely, each operator $X \in \mathcal M_0^\rho$ generates a tangent direction
\begin{equation}
D_X(\rho) = \tfrac{1}{2}(X\rho + \rho X),
\end{equation}
which defines a valid local model
\begin{equation}
\rho_\theta = \rho + \theta D_X(\rho),
\end{equation}
because $\rho_\theta$ remains positive for sufficiently small $\theta$ due to $\rho$ having full rank. Finally, since the eigenspace corresponding to the eigenvalue $1$ of $\mathcal{F}$ is spanned by $\mathbb I$, and $\mathcal F$ is a Hermitian operator with respect to the inner product $\langle\cdot,\cdot\rangle_\rho$, it follows that every eigenvector $B$ of $\mathcal{F}$ with eigenvalue $\lambda \neq 1$ must be orthogonal to $\mathbb I$, and hence satisfy $B \in \mathcal M_0^\rho$. Thus, each such eigenvector corresponds to a valid local model $\rho_\theta$. 
\end{proof}

Having established the geometric and spectral structure of the frame operator, we are now ready to formulate and prove the central result of this work. The state $\rho$ that we choose or are given is interpreted as a \emph{reference state}, representing the \emph{a priori} information available about the system. Since local parameter estimation is performed in a neighborhood of this reference point, the quality of this prior knowledge directly influences how well we can understand the correlation between the measurement and the underlying quantum state. In the absence of any prior knowledge, we take the reference state to be the maximally mixed state, $\rho = \mathbb{I}/d$.

Given such a prior state $\rho$ and an IC-POVM $\mathsf E$, our objective is to determine the \emph{best} and \emph{worst} directions for parameter estimation—namely, the directions in operator space along which the encoded parameter yields the largest or smallest ratio of classical to quantum Fisher information.

\begin{theorem}
Given a full-rank quantum state $\rho$ and an IC-POVM $\mathsf{E}$, 
the ratio between the classical and quantum Fisher information satisfies
\begin{equation}
\lambda_{\min}\!\big(\mathcal{F}\big)
\le
\frac{I_C[\rho,\mathsf{E}]}{I_Q[\rho]}
\le
\lambda_{\max}^{(2)}\!\big(\mathcal{F}\big),
\end{equation}
where $\lambda_{\max}^{(2)}$ and $\lambda_{\min}$ denote, respectively, the second-largest and smallest eigenvalues of the frame operator $\mathcal{F}$.  
The corresponding eigenvectors identify the optimal and least informative directions for local parameter estimation.
\end{theorem}

\begin{proof}
Let $\mathcal{F}$ be the frame operator associated to $\rho$ and $\mathsf E$. Let $L$ denote the symmetric logarithmic derivative (SLD) associated to some local parametric model around $\rho$. Then the ratio of classical to quantum Fisher information can be written as
\begin{equation}
\frac{I_C}{I_Q}
=\frac{\langle L,\mathcal{F} (L)\rangle_\rho}
{\langle L, L\rangle_\rho}.
\end{equation}
By Proposition~\ref{framebounds}, the ordered eigenvalues of the frame operator satisfy
\[
\lambda_1 > \lambda_2 =\lambda_{\max}^{(2)}  \ge \lambda_3 \ge \cdots \ge \lambda_{d^2} = \lambda_{\min} > 0,
\]
where $\lambda_1 = 1$ corresponds to the eigenvector $\mathbb{I}$.
Since $\mathcal F$ is Hermitian, the operator space $\mathcal{M}$ decomposes orthogonally as
\begin{equation}
\mathcal{M} = \mathrm{span}\{\mathbb{I}\} \oplus \mathcal{M}_0^\rho.
\end{equation}
Applying the Rayleigh--Ritz theorem on the invariant subspace $\mathcal{M}_0^\rho$ yields
\begin{align}
\max_{0 \neq X \in \mathcal{M}_0^\rho}
\frac{\langle X,\mathcal{F} (X)\rangle_\rho}{\langle X, X\rangle_\rho}
&= \lambda_{\max}^{(2)}(\mathcal{F}),\\
\min_{0 \neq X \in \mathcal{M}_0^\rho}
\frac{\langle X,\mathcal{F} (X)\rangle_\rho}{\langle X, X\rangle_\rho}
&= \lambda_{\min}(\mathcal{F}).
\end{align}

Since any admissible SLD $L$ lies in $\mathcal{M}_0^\rho$, it follows that
\begin{equation}
\lambda_{\min}(\mathcal{F})
\le
\frac{I_C[\rho, \mathsf{E}]}{I_Q[\rho]}
\le
\lambda_{\max}^{(2)}(\mathcal{F}),
\end{equation}
with equality if and only if $L$ is an eigenvector of $\mathcal{F}$ corresponding to $\lambda_{\max}^{(2)}$ (for the best direction) or $\lambda_{\min}$ (for the worst direction).
\end{proof}

\medskip
\noindent

\begin{figure}[h]
    \centering
    \includegraphics[width=0.5\textwidth]{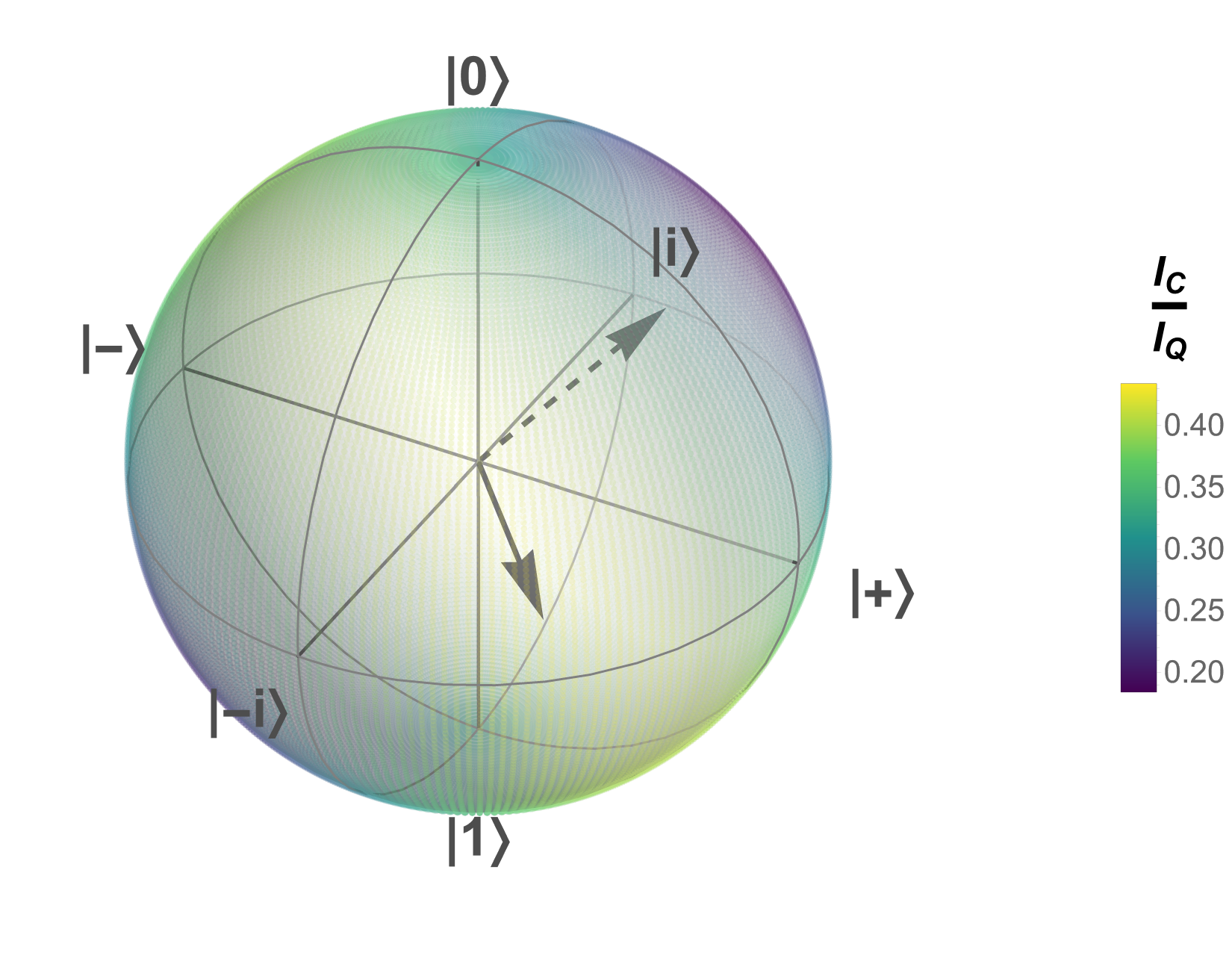}
    \caption{
    In this example, we consider a SIC POVM in a two-dimensional Hilbert space.  
    When no prior information about the quantum state is available (i.e.\ $\rho = \mathbb{I}/2$), the measurement is Fisher-symmetric, so every direction on the Bloch sphere yields the same ratio $I_C/I_Q$. We then choose an arbitrary mixed state ($    \rho = \frac{\mathbb{I}}{2} + 0.3\,\sigma_x + 0.25\,\sigma_y + 0.4\,\sigma_z$),and evaluate the local parameter-estimation performance for all directions on the Bloch sphere, colour-coding the resulting values of $I_C/I_Q$ (dark blue = lower ratio, yellow-green = higher ratio), as shown in the legend.  
    The solid arrow indicates the numerically obtained best estimation direction (maximal $I_C/I_Q$), while the dashed arrow marks the worst direction (minimal $I_C/I_Q$). For comparison, we compute the eigenvectors of the scaled frame operator $\mathcal{F}$ corresponding to its second-largest and smallest eigenvalues, which—according to our theory—identify the optimal and least informative parameter-encoding directions, respectively.  
    The analytical predictions align precisely with the numerical results, demonstrating full consistency with the theoretical characterization.
    }
    \label{Result}
\end{figure}

Hence, for a fixed pair $(\rho, \mathsf{E})$, the \emph{best attainable} ratio $I_C/I_Q$ is achieved when the SLD $L$ aligns with the eigenvector of the frame operator corresponding to the second-largest eigenvalue $\lambda_{\max}^{(2)}$.  
This eigenvector identifies the most informative direction for local parameter estimation, representing the optimal sensitivity achievable by the given measurement under the specified prior information.  
Knowing this full spectral profile allows one to optimize parameter-estimation performance according to experimental goals.  
If the objective is to maximize precision along a specific direction $L$, the POVM can be chosen so that the eigenvector for the second-largest eigenvalue aligns with $L$. Conversely, optimizing for the \emph{worst-case} scenario requires optimizing with respect to the lowest eigenvalue, which characterizes the least informative parameter direction.

In the absence of prior information, we have $\rho=\mathbb{I}/d$, and the state-dependent inner product $\langle A,B\rangle_\rho$ reduces to the Hilbert--Schmidt inner product and the frame operator $\mathcal{F}$ becomes the usual \emph{scaled frame operator} \cite{2025Sai539}.  
When all nontrivial eigenvalues of $\mathcal{F}$ are equal, every operator direction is equally informative, yielding ``Fisher-symmetric'' measurements considered in \cite{2006Sco507,2025Sai539};  
POVMs satisfying this uniformity condition form weighted complex projective $2$-designs \cite{2018zhu404}. To illustrate this regime, we consider a symmetric informationally complete (SIC)-POVM for which the second-largest and smallest eigenvalues of $\mathcal{F}$ coincide, resulting in isotropic parameter estimation performance in the absence of prior information.  
Introducing a nontrivial prior state $\rho$ breaks this symmetry and lifts the spectral degeneracy of the frame operator $\mathcal F$, thereby selecting distinct optimal and least informative parameter-encoding directions, as shown in Fig.~\ref{Result}.

\emph{Conclusion}--- For any IC-POVM, regardless of whether we have prior knowledge of the quantum state or not, the equality between classical and quantum Fisher information can never be achieved. Our main result establishes, in terms of the spectral structure of the associated frame operator, the fundamental limits on how much information on a parameter encoded in a quantum state an IC-POVM can extract. More specifically, the eigenvector corresponding to the second-largest eigenvalue of the frame operator identifies the optimal direction in which the parameter should be encoded to achieve the best estimation precision. Likewise, the smallest eigenvalue of the  frame operator determines the least informative direction and thus the worst-case estimation performance attainable with an IC-POVM. These results provide a complete geometric and operational characterization of information extraction in single-parameter quantum estimation using IC measurements. While IC-POVMs enable estimation along every direction in operator space, our analysis shows that the quality of estimation is determined entirely by the spectral structure of the frame operator. 
This frame-theoretic perspective opens a promising avenue for extensions to the multi-parameter setting, where measurement incompatibility, trade-offs between parameters, and generalized Fisher information bounds become central. Understanding how the spectrum of the frame operator constrains multi-parameter estimation may reveal new optimality criteria and deepen our understanding of the connection between informational completeness and quantum measurement incompatibility.

\begin{acknowledgments}

\end{acknowledgments}

\bibliography{references}  


\appendix

\section{Proof of the nondegeneracy of the maximum eigenvalue of the frame operator}\label{Identity_eigenvalue}

Let $\rho$ be a full-rank density matrix and $\mathsf E=\{E_i\}_{i\in I}$ be an informationally complete POVM. Then we have the corresponding frame operator $\mathcal{F}$, as defined in the main text. We have already shown that all eigenvalues of $\mathcal F$ lie in $[0,1]$, and $\mathbb I$ is an eigenvector of the maximum eigenvalue $1$.

Now, assume that $A$ is also an eigenvector of $\mathcal F$ corresponding to the eigenvalue $1$, that is, $\mathcal F(A)=A$. Then equality holds in \eqref{Fbound}, and hence in \eqref{cauchy} for $E=E_i$, for each $i$. Therefore, the following two conditions are met for each $i$:
\begin{enumerate}
    \item The first inequality is saturated when $(\mathrm{Re}\mathrm{Tr}[\rho E_i A])^2 = |\mathrm{Tr}[\rho E_i A]|^2$.
    \item The Cauchy–Schwarz inequality is saturated when the vectors are proportional to each other: $\sqrt{E_i} A \sqrt{\rho} = \gamma_i \sqrt{E_i} \sqrt{\rho}$ for scalars $\gamma_i$.
\end{enumerate}

Since $\rho$ is full rank and $E_i$ is positive semidefinite, we may multiply on the right by $(\sqrt\rho)^{-1}$, and on the left by $\sqrt{E_i}$, to get
\begin{equation}\label{eq:scalar-multiples}
    E_i A = \gamma_i E_i \quad \text{ for all } i\in I.
\end{equation}
Substituting \eqref{eq:scalar-multiples} into the first condition shows that the $\gamma_i$ are real.
%
Furthermore, expanding the definition of $\mathcal{F}$ in the relation $\langle \rho^{-1}E_j,\mathcal{F}(A) \rangle_\rho = \langle \rho^{-1}E_j,A \rangle_\rho$
gives
\[
\sum_i \frac{\mathrm{Re}[\mathrm{Tr}(\rho \rho^{-1}E_j E_i)]\, \mathrm{Re}[\mathrm{Tr}(\rho E_i A)]}{\mathrm{Re}[\mathrm{Tr}(\rho E_i)]} =  \mathrm{Re}[\mathrm{Tr}(\rho \rho^{-1} E_j A)].
\]
Using \eqref{eq:scalar-multiples} this simplifies to
\[
\gamma_j\,\mathrm{Tr}[E_j] = \sum_i \mathrm{Tr}(E_j E_i)\, \gamma_i.
\]
Note that ${\rm tr}[E_j]>0$ for all $j$, as otherwise $E_j$ would be zero for some $j$. Hence, we may define
$$p(i|j) := \frac{\mathrm{Tr}[E_j E_i]}{\mathrm{Tr}[E_j]}= \frac{\mathrm{Tr}[\sqrt{E_j} E_i\sqrt{E_j}]}{\mathrm{Tr}[E_j]}\ge 0;$$ then $\sum_i p(i|j)=1$, and the above relation reads
\begin{equation}\label{eq:average}
    \gamma_j = \sum_i p(i|j)\,\gamma_i 
    \quad \text{for each }j\in I.
\end{equation}

Let $j$ be such that $\gamma_{\max} = \max\{\gamma_i\}=\gamma_j$. Since \eqref{eq:average} is an average over the probability distribution $p(i|j)$, for each $i$ we have either $\gamma_i=\gamma_{\max}$ or $p(i|j) = 0$. Notice that the latter case means that $E_i$ is orthogonal to $E_j$. We will show that this case cannot occur.

More generally, due to the informational completeness of $\mathsf E$, the outcome set $I$ cannot be partitioned into two nonempty sets $S_1$ and $S_2$ whose effects are orthogonal. To see this, assume for the contrary that such a partition exists. Fix $k\in S_1$, $l\in S_2$, and let $\psi_1$ and $\psi_2$ be (any) eigenvectors of $E_k$ and $E_l$, respectively. Then, for each $i\in S_1$ we have $\|\sqrt{E_i}\psi_2\|^2=\langle \psi_2|E_i|\psi_2\rangle \leq c \  {\rm Tr}[E_iE_l]=0$ for some $c\geq 0$, so $\sqrt{E_i}\psi_2=0$ and hence also $E_i\psi_2=0$. Similarly $E_j\psi_1=0$ for all $j\in S_2$. Now consider the superposition $\ket{\psi_\theta}=(\ket{\psi_1}+e^{i\theta}\ket{\psi_2})/\sqrt{2}$ with an arbitrary phase $\theta$. It follows that $p(E_i|\psi_\theta)=\bra{\psi_\theta}E_i\ket{\psi_\theta}=\bra{\psi_1}E_i\ket{\psi_1}$ if $i\in S_1$ and $p(E_j|\psi_\theta)=\bra{\psi_2}E_j\ket{\psi_2}$ if $j\in S_2$, so $p(E_i|\psi_\theta)$ does not depend on $\theta$ for any $i\in I$. Hence, the POVM does not distinguish the states $\psi_\theta$, which contradicts the informational completeness of $\mathsf E$.

Consequently we cannot have $p(i|j) = 0$ for any $i$, and hence $\gamma_i=\gamma_{\max}$ for all $i$, implying $E_i A = \gamma_{\max}\, E_i$ for all $i$. Summing over $i$ we get $A = \gamma_{\max}\,\mathbb I$. Therefore, up to normalization, the only eigenvector of $\mathcal{F}$ corresponding to the maximum eigenvalue $1$ is the identity operator.

\end{document}